\documentclass[prl,superscriptaddress]{revtex4-2}
\pdfoutput=1
\usepackage{framed,braket,dsfont,subfigure,amsfonts,amssymb,bm,graphicx,float,amsmath,hyperref,amsthm,mathrsfs,soul,xcolor}
\theoremstyle{definition}
\newtheorem{definition}{Definition}
\newtheorem{theorem}{Theorem}
\newtheorem{corollary}{Corollary}
\newtheorem*{lemma}{Lemma}
\newcommand{\Tr}{{\rm Tr}} 
\newcommand{\llvert}{\left|\hspace{-0.5mm}\left|}
\newcommand{\rrvert}{\right|\hspace{-0.5mm}\right|}

\begin{document}
\title{Noise effects on purity and quantum entanglement in terms of physical implementability}
\author{Yuchen Guo}
\affiliation{State Key Laboratory of Low Dimensional Quantum Physics and Department of Physics, Tsinghua University, Beijing 100084, China}
\author{Shuo Yang}
\email{shuoyang@tsinghua.edu.cn}
\affiliation{State Key Laboratory of Low Dimensional Quantum Physics and Department of Physics, Tsinghua University, Beijing 100084, China}
\affiliation{Frontier Science Center for Quantum Information, Beijing 100084, China}
\affiliation{Hefei National Laboratory, Hefei 230088, China}
\begin{abstract}
    Quantum decoherence due to imperfect manipulation of quantum devices is a key issue in the noisy intermediate-scale quantum (NISQ) era.
    Standard analyses in quantum information and quantum computation use error rates to parameterize quantum noise channels.
    However, there is no explicit relation between the decoherence effect induced by a noise channel and its error rate.
    In this work, we propose to characterize the decoherence effect of a noise channel by the physical implementability of its inverse, which is a universal parameter quantifying the difficulty to simulate the noise inverse with accessible quantum channels.
    We establish two concise inequalities connecting the decrease of the state purity and logarithmic negativity after a noise channel to the physical implementability of the noise inverse, which is required to be decomposed as mutually orthogonal unitaries or product channels respectively.
    Our results are numerically demonstrated on several commonly adopted two-qubit noise models.
    We believe that these relations contribute to the theoretical research on the entanglement properties of noise channels and provide guiding principles for quantum circuit design.
\end{abstract}
\maketitle
\section*{Introduction}
Quantum entanglement is an important resource in quantum computers \cite{Horodecki2009, Amico2008}, empowering the establishment of quantum supremacy \cite{Preskill2012, Arute2019}.
The characterization and detection of quantum entanglement \cite{Guhne2009, Wang2022, Liu2022} in physical systems have been the primary concerns of quantum information and computation for decades.
On the other hand, imperfect control of quantum systems in the noisy intermediate-scale quantum (NISQ) era may induce errors \cite{Preskill2018} into quantum circuits composed of unitary gates, which can be described by general quantum channels.

An interesting problem arises: how does the quantum entanglement vary after the implementation of a quantum channel?
Roughly speaking, entangled unitary gates can generate entanglement, while noise channels destroy entanglement.
The balance between these two parts gives the critical point of whether highly-entangled states can be generated in quantum computers, which is a necessary condition to achieve universal quantum computation \cite{Li2019, Zhou2020, Zhang2022}.
Previous studies on entanglement measures of quantum channels \cite{Gour2021, Ballarin2022} give a good response to the former issue, while for the latter, standard analyses on noise often use the error rates to describe the strength of noise effects \cite{Nielsen2009}.
Nevertheless, for different noise models, there is no universal relation between their error rates and decoherence effects.

Here, we try to characterize how much purity and quantum entanglement a noise channel $\mathcal{E}$ destroys with a universal parameter.
We take its inverse $\mathcal{E}^{-1}$ into consideration, which generally is not a physical quantum channel but can recover the quantum entanglement destroyed by $\mathcal{E}$ mathematically since $\mathcal{E}^{-1}\circ\mathcal{E}=\mathcal{I}$.
Intuitively, the harder to implement the noise inverse, the more destructive the noise itself.
Therefore, we believe that for an invertible noise channel $\mathcal{E}$, the physical implementability of its inverse $\mathcal{E}^{-1}$ \cite{Jiang2021}, which represents the sampling cost to implement a linear map \cite{Jiang2021, Takagi2021, Regula2021}, is a prime candidate.
Such a sampling cost measure is constructed from quantum resource theories \cite{Chitambar2019, Vidal1999, Takagi2019}, characterizing the distance between a non-physical linear map and the set of physical quantum channels.
In particular, we establish two concise and universal inequalities bounding the decrease of the state purity and logarithmic negativity \cite{Vidal2002A, Plenio2005, Wang2020B} under noise channels with this measure.
The first inequality is verified by several analytical examples, while the second is numerically demonstrated on four typical two-qubit noise channels.

\section*{Results}
We adopt the Choi operator $\Lambda_{\mathcal{N}}$\cite{Choi1975,Nielsen2009} to represent a quantum linear map $\mathcal{N}$, from which one can construct the output density operator \cite{Supplemental}
\begin{equation}
    \mathcal{N}\left(\rho\right) = \Tr_{\sigma}{\left[\left(\rho^{{\rm T}}\otimes I_{\tau}\right)\Lambda_{\mathcal{N}}\right]}.
\end{equation}
For an invertible completely positive (CP) and trace-preserving (TP) map $\mathcal{T}$, its inverse is Hermitian-preserving (HP) and TP \cite{Cao2021, Jiang2021}.
The sampling cost for implementing an HPTP map $\mathcal{N}$ with the Monte Carlo method is characterized by its physical implementability \cite{Jiang2021, Regula2021}, defined as
\begin{equation}
    \nu\left(\mathcal{N}\right) := \log_2{\min_{\mathcal{T}_i \text{ is CPTP}}{\left\{\left.\sum_{i}|q_i|\right|\mathcal{N} = \sum_i q_i \mathcal{T}_i, q_i\in \mathbb{R}\right\}}}.\label{equ: Definition}
\end{equation}
From the perspective of quantum resource theories, such a quantity measures the distance of an HPTP map $\mathcal{N}$ from the set of CPTP channels $\{\mathcal{T}_i\}$.
Inspired by the fact that the inverse of a quantum channel is still CPTP iff it is unitary \cite{Jiang2021}, we expect that the physical implementability of $\mathcal{E}^{-1}$ can characterize the deviation of the noise channel $\mathcal{E}$ from unitary maps, which preserve purity and are capable of increasing quantum entanglement.

We begin by connecting the decrease of purity to the physical implementability of the noise inverse.
The purity of a quantum state $\rho$ is defined as $\mathcal{P}\left(\rho\right) := \Tr{\left[\rho^2\right]}$ \cite{Nielsen2009}.
We decompose the density operator by a set of Hermitian, complete, and orthonormal basis $\{O_{\alpha}\}$ containing the identity $O_0 = I$ as
\begin{equation}
    \rho = \frac{I + \sum_{i=1}^{d^2-1}r_iO_{i}}{d},
\end{equation}
where $r_{i}$ is defined by the expectation value of $\{O_{i}\}$, i.e., $r_{i} = \Tr{\left[\rho O_{i}\right]}$.
Then the purity is related to the length of the vector $\vec{r}$ by
\begin{equation}
    \mathcal{P}\left(\rho\right) := \Tr{\left[\rho^2\right]} = \Tr{\left[\rho\left(\frac{I+\sum_{i=1}^{d^2-1} r_{i}O_{i}}{d}\right)\right]} = \frac{1+\left|\vec{r}\right|^2}{d},\label{equ: Purity}
\end{equation}
which enables us to calculate the purity change via the transformation of $\vec{r}$.

It was proved in \cite{Jiang2021} that if an HPTP map $\mathcal{N}$ on a $d$-dimensional Hilbert space is the superposition of mutually orthogonal unitaries, such a decomposition is optimal (in the sense of the minimization in Eq. \eqref{equ: Definition}), thus the physical implementability of $\mathcal{N}$ is determined by the trace norm of its Choi operator, i.e., 
\begin{equation}
    2^{\nu\left(\mathcal{N}\right)} = \sum_{i}{\left|q_i\right|} = \frac{\llvert\Lambda_{\mathcal{N}}\rrvert_1}{d},\label{equ: Phys_orthogonal}
\end{equation}
from which we derive the following theorem.
\begin{theorem}
    For a mixed unitary map $\mathcal{N}$ decomposed by a set of mutually orthogonal unitaries, the purity of the input state $\rho_0$ and the output state $\rho$ satisfy
    \begin{equation}
        \log_2 {\left(\frac{\mathcal{P}\left(\rho\right)d - 1}{\mathcal{P}\left(\rho_0\right)d - 1}\right)} \leq 2\nu\left(\mathcal{N}\right),
    \end{equation}
    where $d$ is the dimension of the Hilbert space.\label{the: Purity_Physical}
\end{theorem}
\begin{proof}
    Consider the unitary decomposition $\mathcal{N}\left(\cdot\right) = \sum_i{q_i U_i\left(\cdot\right)U_i^{\dagger}}$, since unitary transformations leave $\left|\vec{r}\right|$ unchanged, we reach
    \begin{equation}
        \left|\vec{r}\left(\rho\right)\right| =\left| \vec{r}\left(\sum_i{q_iU_i\rho_0 U_i^{\dagger}}\right)\right| = \left|\sum_i{q_i \vec{r}\left(U_i\rho_0 U_i^{\dagger}\right)}\right| \leq \sum_i|q_i|\left|\vec{r}\left(U_i\rho_0 U_i^{\dagger}\right)\right| = \sum_i|q_i|\left|\vec{r}\left(\rho_0\right)\right| = 2^{\nu\left(\mathcal{N}\right)}\left|\vec{r}\left(\rho_0\right)\right|,
    \end{equation}
    where we have used Eq. \eqref{equ: Phys_orthogonal}.
    Finally, we build the relationship between the physical implementability of $\mathcal{N}$ and the ratio of purity, i.e., 
    \begin{equation}
        \log_2 {\left(\frac{\mathcal{P}\left(\rho\right)d - 1}{\mathcal{P}\left(\rho_0\right)d - 1}\right)} = \log_2{\frac{\left|\vec{r}\left(\rho\right)\right|^2}{\left|\vec{r}\left(\rho_0\right)\right|^2}}\leq 2\nu\left(\mathcal{N}\right). \qedhere
    \end{equation}
\end{proof}

From Theorem \ref{the: Purity_Physical} we can directly derive the following corollary, which is one of the main results of this work.
\begin{corollary}
    For a noise channel $\mathcal{E}$, if both $\mathcal{E}$ and $\mathcal{E}^{-1}$ are mixed unitary maps decomposed by mutually orthogonal unitaries, the purity of the input state $\rho_0$ and the output state $\rho$ satisfy
    \begin{framed}
    \begin{equation}
            -2 \nu{\left(\mathcal{E}^{-1}\right)}\leq \log_2 {\left(\frac{\mathcal{P}\left(\rho\right)d - 1}{\mathcal{P}\left(\rho_0\right)d - 1}\right)} \leq 2\nu\left(\mathcal{E}\right) = 0,\label{equ: Purity_noise}
        \end{equation}
    \end{framed}
    \noindent where $d$ is the dimension of the Hilbert space.
    \label{coro: purity}
\end{corollary}

The last equality in Eq. \eqref{equ: Purity_noise} follows from the fact that $\nu{\left(\mathcal{T}\right)} = 0$ for any CPTP map $\mathcal{T}$ by definition.
We note that several commonly used noise models, such as the (multiqubit) Pauli noise, depolarizing noise, and dephasing noise, all belong to this category (see Methods).
It can be easily verified that both sides of this inequality can be reached for the single-qubit Pauli noise. 
As for a more generic multiqubit noise, such as the $n$-qubit dephasing noise, we discuss the bounds in Supplementary Information \cite{Supplemental}, where we conclude that the equality may hold when $n = 1$, while the lower bounds can be further tightened for $n\geq 2$.

Now we turn to consider the noise effects on quantum entanglement, which limit the potential power of quantum computers \cite{Zhou2020}.
There are many entanglement measures \cite{Guhne2009} for bipartite mixed states, such as concurrence \cite{Hill1997}, entanglement of formation \cite{Bennett1996,Wootters1998}, entanglement of assistance \cite{Cohen1998}, localizable entanglement \cite{Verstraete2004,Popp2005}, entanglement cost \cite{Vidal2002B,Wang2020A}, etc.
Here we choose the logarithmic negativity \cite{Vidal2002A, Plenio2005} to measure the state entanglement, which characterizes the violation of the well-known positive partial transpose (PPT) criterion \cite{Peres1996}.
For a quantum state $\rho$ on a bipartite system $A\otimes B$, its logarithmic negativity is defined as
\begin{equation}
    E_N\left(\rho\right) := \log_2{\llvert\rho^{{\rm T}_B}\rrvert_1},
\end{equation}
where ${\rm T}_B$ denotes the partial transpose of subsystem $B$.
   
To fully understand this issue, we need to analyze the entanglement property of a quantum channel itself.
Similar to the previous case, we wish to decompose the noise inverse into product quantum channels, inspired by a general property of any entanglement measure $E\left(\rho\right)$, namely that $E\left(\rho\right)$ does not increase under local operations and classical communication (LOCC) \cite{Vedral1997, Vidal2002A}.
The Choi-Jamio\l{}kowski isomorphism \cite{Jamiolkowski1972, Choi1975} between linear maps and density operators motivates us to define separable (entangled) quantum maps \cite{Gour2021}, which is a generalization of separable (entangled) quantum states.
\begin{definition}
    Let $\mathcal{N}$ be an HPTP map on a bipartite system $A\otimes B$.
    We say that $\mathcal{N}$ is separable, if there exist $q_i\in \mathbb{R}$ and product channels $\mathcal{T}_i^A\otimes \mathcal{T}_i^B$ such that
    \begin{equation}
        \mathcal{N}(\cdot) = \sum_i{q_i \left(\mathcal{T}_i^A\otimes \mathcal{T}_i^B\right)\left(\cdot\right)}.\label{equ: generalized}
    \end{equation}
    Otherwise, we call it an entangled map.
    \label{def: Separable}
\end{definition}

We note that not all HPTP maps are separable even if negative coefficients $\{q_i\}$ are allowed in decomposition, which can be proved with the following idea (see Methods for a complete proof).
The Choi operator of any HPTP map can be decomposed with the computational basis
\begin{equation}
    \Lambda_{\mathcal{N}} = \sum_{i=1}^{d}\sum_{j=1}^{d}\ket{i}\hspace{-1mm}\bra{j}^{\sigma}\otimes O_{ij}^{\tau},
\end{equation}
then the TP condition is equivalent to
\begin{equation}
    \Tr{\left[O_{ij}^{\tau}\right]} = \delta_{ij}.
\end{equation}
In this way, if a bipartite HPTP map is separable, we have
\begin{equation}
    \sum_{ijkl}\ket{i}\hspace{-1mm}\bra{j}^{A\sigma}\otimes\ket{k}\hspace{-1mm}\bra{l}^{B\sigma}\otimes O_{ijkl}^{AB\tau} = \sum_m\left(q_m\sum_{ijkl}\ket{i}\hspace{-1mm}\bra{j}^{A\sigma}\otimes\ket{k}\hspace{-1mm}\bra{l}^{B\sigma}\otimes O_{mij}^{A\tau}\otimes O_{mkl}^{B\tau}\right).
\end{equation}
Thus one needs to find the decomposition of $O_{ijkl}^{AB\tau}$ as $O_{ijkl}^{AB\tau} = \sum_m{q_mO_{mij}^{A\tau}\otimes O_{mkl}^{B\tau}}$, which gives $\Tr_{A}\left[O_{ijkl}^{AB\tau}\right]=0$ for $i\neq j$, a much stronger condition than what we have for the original map $\Tr{\left[O_{ijkl}^{AB\tau}\right]}=0$.

To study the influence of noise on the state negativity, we try to connect the partial transpose of output and input states.
The following theorem provides a general bound for the output state negativity concerning the optimal decomposition given in Eq. \eqref{equ: generalized}.
\begin{theorem}
    For a separable HPTP map $\mathcal{N}$ on a bipartite system $A\otimes B$, the logarithmic negativity of the input state $\rho_0$ and the output state $\rho$ satisfy
    \begin{equation}
        E_N\left(\rho\right) - E_N\left(\rho_0\right) \leq \eta\left(\mathcal{N}\right)\equiv\log_2{\min{\left\{\left.\sum_{i}|q_i|\right|\mathcal{N} = \sum_i q_i \mathcal{T}_i^A\otimes\mathcal{T}_i^B, q_i\in \mathbb{R}\right\}}},
    \end{equation}
    where $\mathcal{T}_i^{A(B)}$ are quantum channels on subsystem $A(B)$.\label{the: negativity_Physical}
\end{theorem}
\begin{proof}
    For any decomposition $\mathcal{N} = \sum_i q_i \mathcal{T}_i^A\otimes\mathcal{T}_i^B$, we can bound the trace norm of $\rho^{{\rm T}_B}$ as
    \begin{equation}
        \llvert \rho^{{\rm T}_B} \rrvert_1 = \llvert \sum_i{q_i}{\left[\mathcal{T}_i^A\otimes \mathcal{T}_i^B\left(\rho_0\right)\right]^{{\rm T}_B}}\rrvert_1 \leq \sum_{i}{|q_i|\llvert\left[\mathcal{T}_i^A\otimes \mathcal{T}_i^B\left(\rho_0\right)\right]^{{\rm T}_B}\rrvert_1} \leq \sum_i{|q_i|}\llvert\rho_0^{{\rm T}_B}\rrvert_1,\label{equ: Proof_negativity}\qedhere
    \end{equation}
\end{proof}

The above theorem indicates that $\eta\left(\mathcal{N}\right)$ characterizes the potential of a separable HPTP map $\mathcal{N}$ to increase entanglement, satisfying $\eta\left(\mathcal{N}\right)\geq\nu\left(\mathcal{N}\right)$ by definitions.
Here we focus on a special class of separable HPTP maps with $\eta\left(\mathcal{N}\right)=\nu\left(\mathcal{N}\right)$, i.e., the decomposition by product channels is still optimal in the sense of Eq. \eqref{equ: Definition}.
For example, a noise channel $\mathcal{E}$ generally satisfies that $\eta\left(\mathcal{E}\right) = \nu\left(\mathcal{E}\right)=0$ and cannot increase entanglement.
On the other hand, we believe that the inverse of the noise channel can counteract the noise effect and recover coherence.
It helps us characterize how much entanglement the noise channel destroys.
With the above theorem, we can derive another important conclusion for our study.

\begin{corollary}
    For a noise channel $\mathcal{E}$, if both $\mathcal{E}$ and $\mathcal{E}^{-1}$ are separable with $\eta\left(\mathcal{E}^{-1}\right) = \nu\left(\mathcal{E}^{-1}\right)$, the logarithmic negativity of the input state $\rho_0$ and the output state $\rho$ satisfy
    \begin{framed}
    \begin{equation}
            -\nu{\left(\mathcal{E}^{-1}\right)} \leq \Delta E_N \equiv E_N\left(\rho\right) - E_N\left(\rho_0\right) \leq \nu\left(\mathcal{E}\right) = 0.\label{equ: Negativity_noise}
        \end{equation}
    \end{framed}
    \label{coro: negativity}
\end{corollary}

We note that several typical noise models \cite{Mangini2022} (see Methods) all fall into this category.
In the following, we numerically verify Eq. \eqref{equ: Negativity_noise} with these noise models applied to two-qubit quantum states.
In Fig. \ref{fig: Negativity} (a) we detect the change of logarithmic negativity between two qubits $\left|\Delta E_N\right| \equiv E_N\left(\rho_0\right)-E_N\left(\rho\right)$.
For each type of noise, we randomly sample 10000 two-qubit pure states as input states, i.e., $\rho_0 = \ket{\psi}\hspace{-1mm}\bra{\psi}$, where $E_N\left(\rho_0\right)$ approximately follow a Gaussian distribution within $\left(0, 1\right)$.
It is demonstrated that the bounds given by $\nu{\left(\mathcal{E}^{-1}\right)}$ (solid lines) are about twice as large as the maximal values of $\left|\Delta E_N\right|$.

These bounds are mathematically tight, but maybe not physically.
For example, if we consider the dephasing noise and let the output state be a product state $\rho^{[n]} = \ket{++\cdots +}\hspace{-1mm}\bra{++\cdots +}$ with $E_N\left(\rho^{[n]}\right) = 0$, then the input state can be reconstructed as
\begin{equation}
    \rho_0^{[n]} = \frac{1}{1-\varepsilon}\ket{++\cdots+}\hspace{-1mm}\bra{++\cdots+}-\frac{\varepsilon}{2^n\left(1-\varepsilon\right)}{\sum_{\{i_k=+,-\}}}\ket{i_1i_2\cdots i_n}\hspace{-1mm}\bra{i_1i_2\cdots i_n},
\end{equation}
whose logarithmic negativity is $E_N\left(\rho_0\right) = \log_2{\left[\frac{2^n+\left(2^n-2\right)\varepsilon}{2^n\left(1-\varepsilon\right)}\right]}$.
In this case, we have $\left|\Delta E_N\right| = \nu\left(\mathcal{E}^{-1}\right)$.
However, both $\rho_0$ and $\rho_0^{{\rm T}_B}$ are not positive, which is not a physical situation.
To further demonstrate this point, we provide numerical results in Fig. \ref{fig: Negativity} (b), where we choose mixture of pure states as input states $\rho_0 = \lambda_1\ket{\psi_1}\hspace{-1mm}\bra{\psi_1} + \lambda_2\ket{\psi_2}\hspace{-1mm}\bra{\psi_2}$.
Here $\ket{\psi_1}$ and $\ket{\psi_2}$ are randomly chosen orthonormal two-qubit pure states, while $\lambda_1$ and $\lambda_2$ are randomly chosen from $[-1, 1]$ and then normalized as $\lambda_1+\lambda_2=1$, thus maybe $\rho_0$ are not physical.
It is shown that the variation range of the negativity decrease in Fig. \ref{fig: Negativity} (b) is larger than that in Fig. \ref{fig: Negativity} (a) for each type of noise due to the violation of the positivity restrictions on $\rho_0$ \cite{Supplemental}.
Therefore, we expect that the bounds in Corollary \ref{coro: negativity} may be further tightened for specific noise models with much more careful analyses combined with physical requirements, e.g., the positive conditions on $\rho_0$ and $\rho$.

\begin{figure}
    \subfigure[$\rho_0 = \ket{\psi}\hspace{-1mm}\bra{\psi}$]{\includegraphics[width = 0.49\linewidth]{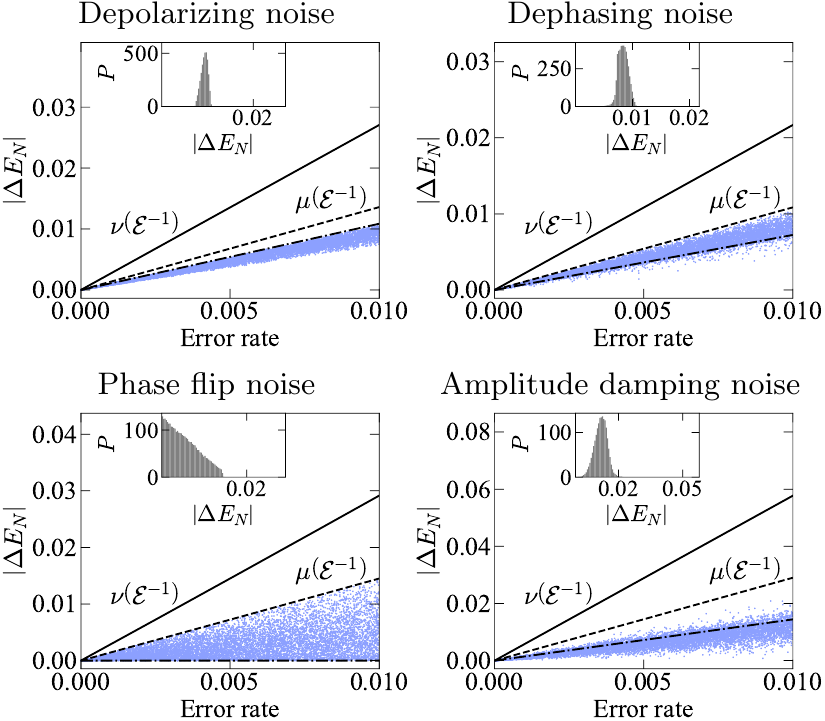}}
    \subfigure[$\rho_0 = \lambda_1\ket{\psi_1}\hspace{-1mm}\bra{\psi_1} + \lambda_2\ket{\psi_2}\hspace{-1mm}\bra{\psi_2}$]{\includegraphics[width = 0.49\textwidth]{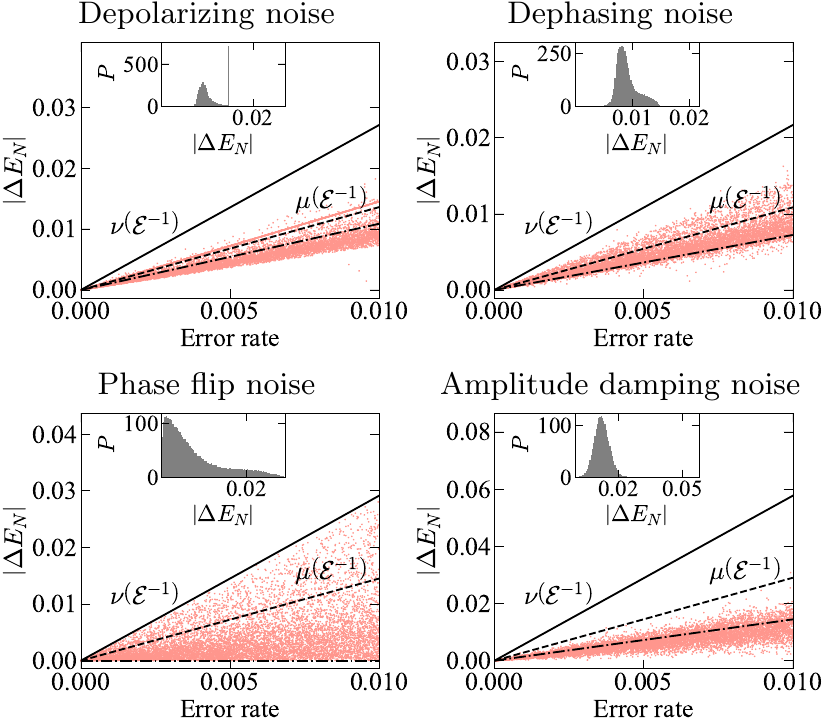}}
	\caption{The change of logarithmic negativity of quantum states $\left|\Delta E_N\right|$.
    The input states are randomly chosen as two-qubit (a) pure states $\rho_0 = \ket{\psi}\hspace{-1mm}\bra{\psi}$ or (b) mixture of pure states $\rho_0 = \lambda_1\ket{\psi_1}\hspace{-1mm}\bra{\psi_1} + \lambda_2\ket{\psi_2}\hspace{-1mm}\bra{\psi_2}$.
    $\lambda_1$ and $\lambda_2$ are randomly chosen from $[-1, 1]$ and normalized as $\lambda_1+\lambda_2=1$.}
    To benchmark, we plot the upper bounds given by the physical implementability of the noise inverse $\nu{\left(\mathcal{E}^{-1}\right)}$ (solid lines), the estimation values given by root-mean-square $\mu{\left(\mathcal{E}^{-1}\right)}$ (dashed lines), and the analytical values derived for the maximally entangled state (dash-dot lines).
    Insets show the probability distributions $P$ of $\left| \Delta E_N\right|$ for error rate $\varepsilon=0.01$.
	\label{fig: Negativity}
\end{figure}

Alternatively, we can use the root-mean-square of $\{q_i\}$ in Eq. \eqref{equ: Proof_negativity} as an estimation value, i.e.,
\begin{equation}
    \llvert \rho^{{\rm T}_B} \rrvert_1 = \llvert \sum_i{q_i}{\left[\mathcal{T}_i^A\otimes \mathcal{T}_i^B\left(\rho_0\right)\right]^{{\rm T}_B}}\rrvert_1 \approx \sqrt{\sum_{i}{|q_i|^2\llvert\left[\mathcal{T}_i^A\otimes \mathcal{T}_i^B\left(\rho_0\right)\right]^{{\rm T}_B}\rrvert_1^2}} \leq \sqrt{\sum_i{|q_i|^2\llvert\rho_0^{{\rm T}_B}\rrvert_1^2}}=\sqrt{\sum_i{|q_i|^2}}\llvert\rho_0^{{\rm T}_B}\rrvert_1.
\end{equation}
If we denote $\mu{\left(\mathcal{N}\right)} = \log_2{\sqrt{\sum_i{q_i^2}}}$, then the above relation can be expressed in a concise form
\begin{equation}
    E_N\left(\rho\right) - E_N\left(\rho_0\right) \lesssim \mu{\left(\mathcal{N}\right)}.
\end{equation}
We simultaneously plot $\mu{\left(\mathcal{E}^{-1}\right)}$ in Fig. \ref{fig: Negativity}, where we surprisingly find that $\mu{\left(\mathcal{E}^{-1}\right)}$ (dashed lines) appear to upper-bound the decrease of the state negativity $\left|\Delta E_N \right|$ for physical situations.

If there is only one positive coefficient in the decomposition $q_1 \approx 1+b\varepsilon$, and all other negative coefficients have the same order of magnitude as $\varepsilon$, (all four types of noise models fall into this category), then to the first order of $\varepsilon$ we have
\begin{equation}
    \begin{aligned}
    2^{\nu\left(\mathcal{N}\right)} &= \sum_i{|q_i|} \approx 1+2b\varepsilon,\\
    2^{\mu\left(\mathcal{N}\right)} &= \sqrt{\sum_i{q_i^2}} \approx 1+b\varepsilon.
    \end{aligned}
\end{equation}
As a result, for small $\varepsilon$ we have $\mu\left(\mathcal{N}\right)\approx \frac{1}{2}\nu\left(\mathcal{N}\right)$, which are demonstrated by the dashed and solid lines in Fig. \ref{fig: Negativity}.

Meanwhile, we note that the probability distributions $P\left(\left|\Delta E_N \right|\right)$ vary with different noise models, which are plotted in the insets of Fig. \ref{fig: Negativity} for a fixed error rate $\varepsilon = 0.01$ without loss of generality.
The total probability (i.e. the area of shadow) is normalized to $1$ for each subfigure.
For example, there exists a sharp peak in Fig. \ref{fig: Negativity} (b) for the depolarizing noise.
This phenomenon can be explained by the special property of the depolarizing noise, which can be written as 
\begin{equation}
    \mathcal{E}^{[n]}\left(\rho_0\right) = \left(1-\varepsilon\right)\rho_0+\frac{\varepsilon}{2^n}I^{[n]}.
\end{equation}
It implies that the spectrum of $\rho^{{\rm T}_B}$ is directly related to that of $\rho_0^{{\rm T}_B}$ by
\begin{equation}
    \lambda_i = \left(1-\varepsilon\right)\lambda^0_i+\frac{\varepsilon}{2^n}
\end{equation}
for $i=1, 2, \cdots 2^n$.
Therefore, if half of the eigenvalues of $\rho_0^{{\rm T}_B}$ are positive and half are negative, and $\varepsilon$ is small which does not change the sign of $\lambda_i$, the trace norm of the output state can be analytically derived, i.e.,
\begin{equation}
    \llvert \rho^{{\rm T}_B} \rrvert_1 = \left(1-\varepsilon\right) \llvert \rho_0^{{\rm T}_B} \rrvert_1,
\end{equation}
which gives the decrease of logarithmic negativity $\left|\Delta E_N\right| = \log_2{\left[\frac{1}{1-\varepsilon}\right]}$.
Meanwhile, we analytically derive $\left|\Delta E_N\right|$ for the maximally entangled state in Supplementary Information \cite{Supplemental} and show them as the dash-dot lines in Fig. \ref{fig: Negativity}.
For the depolarizing noise and the phase flip noise, they serve as the supremum and the infimum of $\left|\Delta E_N\right|$ respectively in Fig. \ref{fig: Negativity} (a), while there is no apparent feature for the amplitude damping and the dephasing noise.
These properties reveal another aspect of the decoherence effect, namely the variety of $\left|\Delta E_N\right|$ for different input states, which we leave for further study.

\section*{Discussion}
In this work, we build two concise and essential inequalities connecting the output state to the input state, where the physical implementability of the noise inverse upper-bounds the decrease of the purity and logarithmic negativity of quantum states.
Central to this is the optimal decomposition of the noise inverse via mutually orthonormal unitaries or product channels, which applies to several commonly-adopted noise models.
Specifically, the former condition is satisfied by the depolarizing, dephasing, and phase flip noise, while the latter one is additionally satisfied by the amplitude damping noise.
These relations imply that the physical implementability of the noise inverse, which is originally proposed to describe the sampling cost for error mitigation and the distance of the noise inverse away from the set of CPTP channels, is a better characterization for the decoherence effect of a noise channel than the commonly used error rate.

Compared with previous works on the entanglement or coherence properties of quantum channels \cite{Gour2021, Wu2022} that describe the potential of a quantum channel to generate entanglement, our study provides a characterization of how destructive a noise channel is, which has applications in benchmarks of quantum hardware.
For instance, when combined with quantum gate set tomography \cite{Merkel2013, Nielsen2021} to obtain a full characterization for noise models of quantum gates, one may estimate whether a quantum device is capable of generating the highly-entangled states required for quantum supremacy instead of directly detecting quantum entanglement in the output state, which is generally a difficult task with exponentially increasing experimental cost \cite{Liu2022}.
Another interesting problem is to apply our results to the tensor network representation of quantum noise \cite{Guo2022}, which naturally captures the correlation of different qubits involved in the noise channel.
Therefore, we believe that our work is a big step forward and opens a new avenue toward the theoretical and experimental research of noise channels from the entanglement perspective.

\section*{Methods}
\subsection{Physical implementability}
The core physical quantity in our work, the physical implementability of an HPTP map, derives from the quasi-probability method \cite{Temme2017, Endo2018} for quantum error mitigation and its variants \cite{Huo2018, Cao2021, Guo2022, Piveteau2022}, which involve the simulation of the inverse of noise channels with physically implementable quantum channels.
The sampling cost for implementing an HPTP map $\mathcal{N}$ is characterized by its physical implementability \cite{Jiang2021,Regula2021}, defined as
\begin{equation}
    \nu\left(\mathcal{N}\right) := \log_2{\min_{\mathcal{T}_i \text{ is CPTP}}{\left\{\left.\sum_{i}|q_i|\right|\mathcal{N} = \sum_i q_i \mathcal{T}_i, q_i\in \mathbb{R}\right\}}}.
\end{equation}
During the proof of our main results, we take advantage of the fact that the state purity remains unchanged under unitary channels, while the state negativity is non-increasing under product channels.
It indicates that if the optimal decomposition of an HPTP map gives unitary (product) channels, we can bound the increase of purity (negativity) with its physical implementability.
On the other hand, for a mixed unitary map, if it is decomposed by a set of mutually orthogonal unitary channels, such a decomposition is optimal.

Numerically, the physical implementability can be calculated via semidefinite programming (SDP) \cite{Jiang2021}.
Alternatively, we provide the upper and lower bounds for the physical implementability in terms of the maximum and minimum eigenvalues of the Choi matrix in Supplementary Information \cite{Supplemental}, which may inspire efficient estimation methods for the physical implementability with numerical approaches, such as the tensor network representation \cite{Verstraete2008, Orus2014, Cirac2021} of a general noise channel \cite{Guo2022}, instead of solving the entire optimization problem.

\subsection{Noise models and physical implementability}
In the following, we summarize four commonly used noise models and the physical implementability of their inverse, which enables the application of our results on these noise channels.
We will show that Corollary \ref{coro: purity} applies to the multiqubit Pauli noise, depolarizing noise, and dephasing noise, while Corollary \ref{coro: negativity} holds for the multiqubit amplitude damping noise additionally.
\subsubsection{The multiqubit Pauli noise}
We define the multiqubit Pauli noise as
\begin{equation}
    \mathcal{E}^{[n]}\left(\rho^{[n]}\right) = \left(1-\varepsilon\right) \rho^{[n]} + \varepsilon \left(\bigotimes_{i=1}^n{\sigma_{\alpha_i}^i}\right)\rho^{[n]}\left(\bigotimes_{i=1}^n{\sigma_{\alpha_i}^i}\right)^{\dagger},
\end{equation}
where $\sigma_{\alpha_i}^i$ represents the Pauli matrix $\sigma_{\alpha_i}$ applied on the $i$-th site.
The inverse of this noise is analytically derived as
\begin{equation}
    {\mathcal{E}^{n}}^{-1}\left(\rho^{[n]}\right) = \frac{1-\varepsilon}{1-2\varepsilon} \rho^{[n]} - \frac{\varepsilon}{1-2\varepsilon} \left(\bigotimes_{i=1}^n{\sigma_{\alpha_i}^i}\right)\rho^{[n]}\left(\bigotimes_{i=1}^n{\sigma_{\alpha_i}^i}\right)^{\dagger},
\end{equation}
which provides the optimal decomposition, with the physical implementability being $\nu{\left({\mathcal{E}^{[n]}}^{-1}\right)} = \log_2{\left[\frac{1}{1-2\varepsilon}\right]}$.
Each term in the decomposition is a unitary and product channel, hence both Corollary $\ref{coro: purity}$ and $\ref{coro: negativity}$ hold.
The two-qubit phase flip noise in Fig. \ref{fig: Negativity}(c) corresponds to taking $\alpha_1=\alpha_2=3$ here, i.e., $\mathcal{E}^{[2]}\left(\rho^{[2]}\right) = \left(1-\varepsilon\right)\rho^{[2]}+\varepsilon\left(\sigma^1_z\otimes\sigma^2_z\right)\rho^{[2]}{\left(\sigma^1_z\otimes\sigma^2_z\right)}^{\dagger}$.

\subsubsection{The multiqubit depolarizing noise}
The $n$-qubit depolarizing noise is defined as
\begin{equation}
    \begin{aligned}
        \mathcal{E}^{[n]}\left(\rho^{[n]}\right) &= \left(1-\varepsilon\right)\rho^{[n]} + \frac{\varepsilon}{2^n}I^{[n]}\\
        &= \left(1-\varepsilon\right)\rho^{[n]} + \frac{\varepsilon}{4^n}\sum_{\{\alpha_i\}}\left[\left(\bigotimes_{i=1}^n{\sigma_{\alpha_i}^i}\right)\rho^{[n]}\left(\bigotimes_{i=1}^n{\sigma_{\alpha_i}^i}\right)^{\dagger}\right],
    \end{aligned}
\end{equation}
where the indices $\alpha_i$ are summed from 0 to 3.
Its inverse can be directly calculated with the first equality
\begin{equation}
    \begin{aligned}
        {\mathcal{E}^{[n]}}^{-1}\left(\rho^{[n]}\right) &= \frac{1}{1-\varepsilon}\rho^{[n]} - \frac{\varepsilon}{\left(1-\varepsilon\right)2^n} I^{[n]}\\
        &= \frac{1}{1-\varepsilon}\rho^{[n]} - \frac{\varepsilon}{\left(1-\varepsilon\right)4^n}\sum_{\{\alpha_i\}}\left[\left(\bigotimes_{i=1}^n{\sigma_{\alpha_i}^i}\right)\rho^{[n]}\left(\bigotimes_{i=1}^n{\sigma_{\alpha_i}^i}\right)^{\dagger}\right].\label{equ: Depolarizing_inverse}
    \end{aligned}
\end{equation}
which is decomposed by unitary and product channels.
The unitaries in the decomposition of $\mathcal{E}^{[n]}$ (and ${\mathcal{E}^{[n]}}^{-1}$) are also mutually orthogonal
\begin{equation}
    \Tr{\left[\left(\bigotimes_{i=1}^n{\sigma_{\alpha_i}^i}\right)^{\dagger}\left(\bigotimes_{i=1}^n{\sigma_{\beta_i}^i}\right)\right]} = 2^n \prod_{i=1}^{n}{\delta_{\alpha_i\beta_i}},
\end{equation}
allowing us to analytically calculate the physical implementability as
\begin{equation}
    \nu\left({\mathcal{E}^{[n]}}^{-1}\right) = \log_2\left[\left(\frac{1}{1-\varepsilon} - \frac{\varepsilon}{\left(1-\varepsilon\right)4^n}\right) + \left(4^n - 1\right)\frac{\varepsilon}{\left(1-\varepsilon\right)4^n}\right] = \log_2\left[\frac{1+ \left(1 - \frac{2}{4^n}\right)\varepsilon}{1-\varepsilon}\right].\label{Physical_depolarizing}
\end{equation}

\subsubsection{The multiqubit dephasing noise}
The $n$-qubit dephasing noise is defined as
\begin{equation}
    \mathcal{E}^{[n]}\left(\rho^{[n]}\right) = \left(1-\varepsilon\right)\rho^{[n]} + \frac{\varepsilon}{2^n}\sum_{\{\alpha_i \in \{0, 3\}\}}\left[\left(\bigotimes_{i=1}^n{\sigma_{\alpha_i}^i}\right)\rho^{[n]}\left(\bigotimes_{i=1}^n{\sigma_{\alpha_i}^i}\right)^{\dagger}\right],\label{equ: Dephasing}
\end{equation}
where the summation only contains $\sigma_0$ and $\sigma_3$.
We assume that the inverse of $\mathcal{E}^{[n]}$ takes a similar form
\begin{equation}
    {\mathcal{E}^{[n]}}^{-1}\left(\rho^{[n]}\right) = A\rho^{[n]} - B\sum_{\{\alpha_i \in \{0, 3\}\}}\left[\left(\bigotimes_{i=1}^n{\sigma_{\alpha_i}^i}\right)\rho^{[n]}\left(\bigotimes_{i=1}^n{\sigma_{\alpha_i}^i}\right)^{\dagger}\right]
\end{equation}
and solve the undetermined coefficients $A$ and $B$.
Finally, we obtain
\begin{equation}
    {\mathcal{E}^{[n]}}^{-1}\left(\rho^{[n]}\right) = \frac{1}{1-\varepsilon}\rho^{[n]} - \frac{\varepsilon}{\left(1-\varepsilon\right)2^n}\sum_{\{\alpha_i \in \{0, 3\}\}}\left[\left(\bigotimes_{i=1}^n{\sigma_{\alpha_i}^i}\right)\rho^{[n]}\left(\bigotimes_{i=1}^n{\sigma_{\alpha_i}^i}\right)^{\dagger}\right],
\end{equation}
where each term is product channels and mutually orthonormal unitaries.
Similarly, we can derive the physical implementability of the noise inverse
\begin{equation}
    \nu\left({\mathcal{E}^{[n]}}^{-1}\right) = \log_2\left[\left(\frac{1}{1-\varepsilon} - \frac{\varepsilon}{\left(1-\varepsilon\right)2^n}\right) + \left(2^n - 1\right)\frac{\varepsilon}{\left(1-\varepsilon\right)2^n}\right] = \log_2\left[\frac{1+ \left(1 - \frac{2}{2^n}\right)\varepsilon}{1-\varepsilon}\right].\label{equ: Physical_dephasing}
\end{equation}

\subsubsection{The amplitude damping noise}
The amplitude damping noise is commonly adopted to describe the loss of photons in quantum systems \cite{Nielsen2009}, which is defined with the Kraus operator $E_0 = \ket{0}\hspace{-1mm}\bra{0} + \sqrt{1-\varepsilon}\ket{1}\hspace{-1mm}\bra{1}$ and $E_1 = \sqrt{\varepsilon}\ket{0}\hspace{-1mm}\bra{1}$ with the operator-sum representation
\begin{equation}
    \mathcal{E}\left(\rho\right) = E_0\rho E_0^{\dagger} + E_1\rho E_1^{\dagger}.
\end{equation}
The physical implementability of $\mathcal{E}^{-1}$ was analytically studied in \cite{Jiang2021}, given by $\nu\left(\mathcal{E}^{-1}\right) = \log_2\left[\frac{1+\varepsilon}{1-\varepsilon}\right]$.
The multiqubit amplitude damping noise is just defined as the tensor product of single-qubit noise channels, i.e., 
\begin{equation}
    \mathcal{E}^{[n]} = \bigotimes_{i=1}^n\mathcal{E}^i
\end{equation}
with the physical implementability of the noise inverse being $\nu{\left({\mathcal{E}^{[n]}}^{-1}\right)} = n\log_2\left[\frac{1+\varepsilon}{1-\varepsilon}\right]$.
Therefore, Corollary \ref{coro: negativity} applies to this noise model.

\subsection{Separable HPTP maps}
We notice that any bipartite quantum state, separable or entangled, can be decomposed as the superposition of product states $\rho = \sum q_i\rho_i^A\otimes \rho_i^B$ if negative coefficients $q_i$ are allowed.
However, this property does not hold for quantum maps.
In this section, we provide detailed proof of the existence of HPTP maps that do not fall into the class of separable HPTP maps defined in Definition \ref{def: Separable}.

Consider a general HPTP map $\mathcal{N}$, whose Choi operator satisfies the HP condition
\begin{equation}
    \Lambda_{\mathcal{N}}^{\dagger} = \Lambda_{\mathcal{N}}
\end{equation}
and the TP condition
\begin{equation}
    \Tr{\left[\Lambda_{\mathcal{N}}\right]}_{\tau} = I_{\sigma}.
\end{equation}
We now decompose the physical part of $\Lambda_{\mathcal{N}}$ with the computational basis
\begin{equation}
    \Lambda_{\mathcal{N}} = \sum_{i=1}^{d}\sum_{j=1}^{d}\ket{i}\hspace{-1mm}\bra{j}^{\sigma}\otimes O_{ij}^{\tau},
\end{equation}
then the HP and TP conditions are equivalent to
\begin{equation}
    O_{ij}^{\tau\dagger} = O_{ji}^{\tau},\\
    \Tr{\left[O_{ij}^{\tau}\right]} = \delta_{ij}.
\end{equation}

Now we try to decompose an arbitrary bipartite HPTP map, whose Choi operator is decomposed as
\begin{equation}
    \Lambda_{\mathcal{N}^{AB}} = \sum_{ijkl}\ket{i}\hspace{-1mm}\bra{j}^{A\sigma}\otimes\ket{k}\hspace{-1mm}\bra{l}^{B\sigma}\otimes O_{ijkl}^{AB\tau},
\end{equation}
satisfying that
\begin{equation}
    \Tr{\left[O_{ijkl}^{AB\tau}\right]} = \delta_{ij}\delta_{kl}.
\end{equation}
If it can be decomposed as superpositions of product channels, i.e.,
\begin{equation}
    \Lambda_{\mathcal{N}^{AB}} = \sum_m{q_m \left(\Lambda_{\mathcal{T}_m^A}\otimes \Lambda_{\mathcal{T}_m^B}\right)}.
\end{equation}
where product channels are further decomposed as
\begin{equation}
    \Lambda_{\mathcal{T}_m^A} = \sum_{ij}\ket{i}\hspace{-1mm}\bra{j}^{A\sigma}\otimes O_{mij}^{A\tau},\\
    \Lambda_{\mathcal{T}_m^B} = \sum_{kl}\ket{k}\hspace{-1mm}\bra{l}^{B\sigma}\otimes O_{mkl}^{B\tau}.
\end{equation}
We reach
\begin{equation}
    \sum_{ijkl}\ket{i}\hspace{-1mm}\bra{j}^{A\sigma}\otimes\ket{k}\hspace{-1mm}\bra{l}^{B\sigma}\otimes O_{ijkl}^{AB\tau} = \sum_m\left(q_m\sum_{ijkl}\ket{i}\hspace{-1mm}\bra{j}^{A\sigma}\otimes\ket{k}\hspace{-1mm}\bra{l}^{B\sigma}\otimes O_{mij}^{A\tau}\otimes O_{mkl}^{B\tau}\right).
\end{equation}
Therefore, we need to find the decomposition of $O_{ijkl}^{AB\tau}$
\begin{equation}
    O_{ijkl}^{AB\tau} = \sum_m{q_mO_{mij}^{A\tau}\otimes O_{mkl}^{B\tau}},\label{equ: General_decom}
\end{equation}
satisfying that
\begin{equation}
    \Tr{\left[O_{mij}^{A\tau}\right]} = \delta_{ij},\\
    \Tr{\left[O_{mkl}^{B\tau}\right]} = \delta_{kl}.
\end{equation}
Such a decomposition cannot always be found.
For example, for $i\neq j$, the above conditions give $\Tr{\left[O_{mij}^{A\tau}\right]}=0$.
Then after taking the partial trace of Eq. \eqref{equ: General_decom}, we require that
\begin{equation}
    \Tr_{A}\left[O_{ijkl}^{AB\tau}\right]=0,
\end{equation}
which is a much stronger condition than what we have for the original map $\Tr{\left[O_{ijkl}^{AB\tau}\right]}=0$.
We note that some commonly encountered two-qubit quantum gates, such as the CNOT gate and the SWAP gate, are not separable due to the above argument.

\subsection{Partial transpose of linear maps}
When evaluating the change of logarithmic negativity after the implementation of a noise channel, one has to connect the partial transpose of the output state and the input state.
In the following lemma, we prove that such a connection can be built in terms of the partial transpose of the linear map.
\begin{lemma}
    For a linear map $N$ on a bipartite system $A\otimes B$, the partial transpose of the output operator satisfies
    \begin{equation}
        \rho^{{\rm T}_B} = \mathcal{N}^{{\rm T}_B}\left(\rho_0^{{\rm T}_B}\right),
    \end{equation}
    where the (partial) transpose of a linear map $\mathcal{N}$ is defined by the (partial) transpose of its Choi operator, i.e.,
    \begin{equation}
            \Lambda_{\mathcal{N}^{{\rm T_{(B)}}}} = \Lambda_{\mathcal{N}}^{{\rm T_{(B)}}}.
    \end{equation}
\end{lemma}
\begin{proof}
    We first take the Schmidt decomposition of operators
    \begin{equation}
        \rho_0 = \sum_k {\alpha_k \rho_k^A\otimes \rho_k^B},\\
        \Lambda_{\mathcal{N}} = \sum_k{\beta_k\Lambda_k^A\otimes \Lambda_k^B},
    \end{equation}
    then the partially transposed output operator is calculated in terms of the Choi operator as
    \begin{equation}
        \begin{aligned}
            \rho^{{\rm T}_B} &= \Tr_{\sigma}{\left[\left(\rho_0^{{\rm T}}\otimes I_{\tau}\right)\Lambda_{\mathcal{N}}\right]}^{{\rm T}_B}\\
            &= \Tr_{\sigma}{\left[\sum_{i,j} {\alpha_i \beta_j\left(\left({\rho_i^A}^{{\rm T}}\otimes I_{\tau}^A\right)\otimes \left({\rho_i^B}^{{\rm T}}\otimes I_{\tau}^B\right)\right)\left(\Lambda_j^A\otimes \Lambda_j^B\right)}\right]}^{{\rm T}_B}\\
            &= \Tr_{\sigma}{\left[\sum_{i,j} {\alpha_i \beta_j\left(\left({\rho_i^A}^{{\rm T}}\otimes I_{\tau}^A\right)\otimes {\Lambda_j^B}^{{\rm T}}\right)\left(\Lambda_j^A\otimes \left(\rho_i^B\otimes I_{\tau}^B\right)\right)}\right]}\\
            &= \Tr_{\sigma}{\left[\sum_{i,j} {\alpha_i \beta_j\left(\left({\rho_i^A}^{{\rm T}}\otimes I_{\tau}^A\right)\otimes \left(\rho_i^B\otimes I_{\tau}^B\right)\right)\left(\Lambda_j^A\otimes {\Lambda_j^B}^{{\rm T}}\right)}\right]}\\
            &= \Tr_{\sigma}{\left[\left(\rho_0^{{\rm T}_A}\otimes I_{\tau}\right)\Lambda_{\mathcal{N}}^{{\rm T}_B}\right]}= \Tr_{\sigma}{\left[\left(\left(\rho_0^{{\rm T}_B}\right)^{\rm T}\otimes I_{\tau}\right)\Lambda_{\mathcal{N}^{{\rm T}_B}}\right]}\\
            &= \mathcal{N}^{{\rm T}_B}\left(\rho_0^{{\rm T}_B}\right).
        \end{aligned}
    \end{equation}
    For the third equality, the partial transpose will introduce a permutation between ${\rho_i^B}^{{\rm T}}\otimes I_{\tau}^B$ and $\Lambda_j^B$, while the permutation in the fourth equality comes from the partial trace of $\sigma$ part and the identity of $\tau$ part.
\end{proof}

\section*{Data Availability}
The datasets generated and analyzed during the current study are available from the corresponding author upon reasonable request.

\section*{acknowledgments}
    We thank Yanzhen Wang for helpful discussions.
    This work is supported by the National Natural Science Foundation of China (NSFC) (Grant No. 12174214 and No. 92065205), the National Key R\&D Program of China (Grant No. 2018YFA0306504), and the Innovation Program for Quantum Science and Technology (Grant No. 2021ZD0302100).

\section*{Author Contributions}
Y.G. conceived, designed, and performed the numerical experiments.
Y.G. and S.Y. analyzed the data and wrote the paper.
S.Y. contributed analysis tools.

\section*{Competing Interests}
The authors declare no competing interests.

\bibliography{ref}

\appendix
\renewcommand{\thesection}{S-\arabic{section}} \renewcommand{\theequation}{S%
\arabic{equation}} \setcounter{equation}{0} \renewcommand{\thefigure}{S%
\arabic{figure}} \setcounter{figure}{0}
\section*{Supplementary Information}
In this Supplementary Information, we provide more details on the Choi operator representation, orthogonally mixed unitary noise, an example of purity change, numerical results for mixed states, decoherence effects on the maximally entangled state, and an estimation method for physical implementability.
\subsection{Choi operator representation}
We denote the Hilbert space of a quantum system as $\mathcal{H}$, where the space of linear operators is labeled as $\mathscr{L}(\mathcal{H})$.
A quantum state is represented by a density operator $\rho \in \mathscr{L}(\mathcal{H})$, which is a Hermitian ($\rho = \rho^{\dagger}$) and positive semidefinite ($\rho \geq 0$) operator with unit trace ($\Tr{[\rho]} = 1$).

A quantum channel $\mathcal{T}$ is a linear map between two operator spaces $\mathcal{T}: \mathscr{L}(\mathcal{H})\rightarrow \mathscr{L}(\mathcal{H})$ with the completely positive (CP) and trace-preserving (TP) conditions.
We say that a linear map $\mathcal{N}$ is
\begin{itemize}
    \item trace-preserving (TP), if $\Tr{[\mathcal{N}(O)]} = \Tr{[O]}$ for $\forall O\in \mathscr{L}(\mathcal{H})$,
    \item Hermitian-preserving (HP), if $\mathcal{N}(O)^{\dagger} = \mathcal{N}(O)$ for $\forall O\in \mathscr{L}(\mathcal{H})$ with $O^{\dagger} = O$,
    \item positive, if $\mathcal{N}(O) \geq 0$ for $\forall O\in \mathscr{L}(\mathcal{H})$ with $O \geq 0$,
    \item completely positive (CP), if $\mathcal{I}_{A^{\prime}}\otimes \mathcal{N}_A$ is positive for an arbitrary ancillary system $A^{\prime}$.
\end{itemize}

The most commonly adopted representation of a quantum channel $\mathcal{T}$ is the Choi operator \cite{Choi1975}, which is defined on a joint system as
    
\vspace{-3mm}
\begin{align}
    \Lambda_{\mathcal{T}} := \left(\mathcal{I}_{\sigma}\otimes \mathcal{T}_{\tau}\right)\left(d\ket{\Phi}\hspace{-1mm}\bra{\Phi}_{\sigma\tau}\right), \label{equ:Choi_operator_sup}
\end{align}
where $d$ is the dimension of the Hilbert space.
$\tau$ and $\sigma$ are the indices of the physical system and the ancillary system (which is isomorphic to the physical system) respectively.
$\ket{\Phi}$ is the maximally entangled state 
    
\vspace{-3mm}
\begin{align}
    \ket{\Phi}_{\sigma\tau} = \frac{1}{\sqrt{d}}\sum_{i=0}^{d-1}\ket{i_{\sigma}i_{\tau}}\label{equ:Maximally_entangled}
\end{align}
with $\left\{\ket{i}_{\tau}\right\}$ and $\left\{\ket{i}_{\sigma}\right\}$ being the orthonormal basis of the physical system and the ancillary system respectively.
On the same basis, the matrix elements of the Choi operator are
    
\vspace{-3mm}
\begin{align}
    \bra{i_{\sigma}k_{\tau}}\Lambda_{\mathcal{T}}\ket{j_{\sigma}l_{\tau}} = \bra{k}\mathcal{T}\left(\ket{i}\hspace{-1mm}\bra{j}\right)\ket{l},\label{equ:Choi_matrix}
\end{align}
which is just the matrix representation of map $\mathcal{T}$ in an orthonormal basis of $\mathscr{L}(\mathcal{H})$ with reshuffled physical indices.
In this sense, we can consider $\sigma$ and $\tau$ as the input and output indices respectively.
Then the implementation of $\mathcal{T}$ can be calculated in terms of the Choi operator as
    
\vspace{-3mm}
\begin{align}
    \mathcal{T}\left(\rho\right) = \Tr_{\sigma}{\left[\left(\rho^{{\rm T}}\otimes I_{\tau}\right)\Lambda_{\mathcal{T}}\right]}.
\end{align}

The Choi representation can be directly generalized to any linear map $\mathcal{N}$ beyond CPTP ones with the same definition as Eq. \eqref{equ:Choi_operator_sup}.
Due to the Choi-Jamio\l{}kowski isomorphism, the conditions mentioned above for a linear map $\mathcal{N}$ are guaranteed by the following properties of $\Lambda_{\mathcal{N}}$ \cite{Jamiolkowski1972}
\begin{itemize}
    \item $\mathcal{N}$ is TP $\Longleftrightarrow$ $\Tr_{\tau}[\Lambda_{\mathcal{N}}] = I_{\sigma}$.
    \item $\mathcal{N}$ is HP $\Longleftrightarrow$ $\Lambda_{\mathcal{N}}^{\dagger} = \Lambda_{\mathcal{N}}$.
    \item $\mathcal{N}$ is CP $\Longleftrightarrow$ $\Lambda_{\mathcal{N}} \geq 0$.
\end{itemize}

\subsection{Orthogonal mixed unitary noise}
We say that a noise channel $\mathcal{E}$ is an orthogonal mixed unitary noise if it satisfies the condition of Corollary 1 in the main text.
In the following lemmas, We prove that the composition of orthogonally mixed unitary noise channels with unitary channels or the tensor product of orthogonally mixed unitary noise channels is also an orthogonally mixed unitary noise.
\begin{lemma}
    For an orthogonally mixed unitary noise $\mathcal{E}$, $\mathcal{U}\circ \mathcal{E}\circ \mathcal{V}$ is also an orthogonally mixed unitary noise for any unitary channels $\mathcal{U}\left(\cdot\right) = U\left(\cdot\right)U^{\dagger}$ and $\mathcal{V}\left(\cdot\right) = V\left(\cdot\right)V^{\dagger}$.
\end{lemma}
\begin{proof}
    Consider the orthogonal decomposition $\mathcal{E}\left(\cdot\right) = \sum_{i}{q_i U_i\left(\cdot\right)U_i^{\dagger}}$, we can decompose
    
\vspace{-3mm}
\begin{align}
        \mathcal{U}\circ\mathcal{E}\circ \mathcal{V}\left(\cdot\right) = \sum_{i}{q_i UU_iV\left(\cdot\right)\left(UU_iV\right)^{\dagger}}
    \end{align}
    with orthogonal unitaries $\{UU_iV\}$ satisfying
    
\vspace{-3mm}
\begin{align}
        \Tr{\left[\left(UU_iV\right)^{\dagger}\left(UU_jV\right)\right]} = \Tr{\left[V^{\dagger}U_i^{\dagger}U^{\dagger}UU_jV\right]} = \Tr{\left[U_i^{\dagger}U_j\right]} = d\delta_{ij},
    \end{align}
    where $d$ is the dimension of the Hilbert space.
    Similarly we can prove that the inverse $\left(\mathcal{U}\circ\mathcal{E}\circ\mathcal{V}\right)^{-1} = \mathcal{V}^{-1}\circ \mathcal{E}^{-1}\circ \mathcal{U}^{-1}$ is also a mixed unitary map decomposed by a set of mutually orthogonal unitaries.
\vspace{-2.5mm}\end{proof}
\begin{lemma}
    For two orthogonally mixed unitary noise channels $\mathcal{E}^A$ and $\mathcal{E}^B$, $\mathcal{E}^A\otimes \mathcal{E}^B$ is also an orthogonally mixed unitary noise.
\end{lemma}
\begin{proof}
    Consider the orthogonal decompositions $\mathcal{E}^X\left(\cdot\right) = \sum_{i}{q_i^X U_i^X\left(\cdot\right){U_i^X}^{\dagger}}$ ($X = A, B$), one can derive
    
\vspace{-3mm}
\begin{align}
        \mathcal{E}^A\otimes\mathcal{E}^B = \sum_{ij}{q_i^Aq_j^B U_i^A\otimes U_j^B\left(\cdot\right){U_i^A}^{\dagger}\otimes {U_j^B}^{\dagger}}.
    \end{align}
    Obviously this decomposition also satisfies the orthogonal condition
    
\vspace{-3mm}
\begin{align}
        \Tr{\left[\left({U_i^A}^{\dagger}\otimes {U_j^B}^{\dagger}\right)\left({U_k^A}\otimes {U_l^B}\right)\right]} = \Tr{\left[{U_i^A}^{\dagger}{U_k^A}\right]}\Tr{\left[{U_j^B}^{\dagger}{U_l^B}\right]} = d^Ad^B\delta_{ik}\delta_{jl},
    \end{align}
    where $d^A$ and $d^B$ are the dimensions of the Hilbert spaces $\mathcal{H}_A$ and $\mathcal{H}_B$ respectively.
    It can also be easily verified that ${\mathcal{E}^{A}}^{-1}\otimes{\mathcal{E}^{B}}^{-1}$ is a mixed unitary map decomposed by a set of mutually orthogonal unitaries.
\vspace{-2.5mm}\end{proof}

\subsection{An example of purity change: Dephasing noise}
We consider a quantum state $\rho_0^{[n]}$ undergoing an $n$-qubit dephasing noise defined in Eq. (33) (see Methods).
We measure the expectation values of $O_{\{\alpha_i\}} = \otimes_{i=1}^{n}{\sigma_{\alpha_i}^i}$, which form a set of orthogonal basis of the operator space.

\vspace{-3mm}
\begin{align}
    \begin{aligned}
        \braket{O_{\{\alpha_i\}}} &= \Tr{\left[\rho^{[n]} O_{\{\alpha_i\}}\right]} = \Tr{\left[\mathcal{E}^{[n]}\left(\rho_0^{[n]}\right)O_{\{\alpha_i\}}\right]}\\
        &= \left(1-\varepsilon\right)\Tr{\left[\rho_0^{[n]}O_{\{\alpha_i\}}\right]} + \frac{\varepsilon}{2^n}\sum_{\{\beta_i \in \{0, 3\}\}}\Tr{\left[\left(\bigotimes_{i=1}^n{\sigma_{\beta_i}^i}\right)\rho_0^{[n]}\left(\bigotimes_{i=1}^n{\sigma_{\beta_i}^i}\right)^{\dagger}\left(\bigotimes_{i=1}^n{\sigma_{\alpha_i}^i}\right)\right]}\\
        &= \left(1-\varepsilon\right)\Tr{\left[\rho_0^{[n]}O_{\{\alpha_i\}}\right]} + \frac{\varepsilon}{2^n}\sum_{\{\beta_i \in \{0, 3\}\}}\Tr{\left[\rho_0^{[n]}\left(\bigotimes_{i=1}^n{\sigma_{\beta_i}^i}\right)^{\dagger}\left(\bigotimes_{i=1}^n{\sigma_{\alpha_i}^i}\right)\left(\bigotimes_{i=1}^n{\sigma_{\beta_i}^i}\right)\right]}\\
        &= \left(1-\varepsilon\right)\Tr{\left[\rho_0^{[n]}O_{\{\alpha_i\}}\right]} + \frac{\varepsilon}{2^n}\sum_{\{\beta_i \in \{0, 3\}\}}\Tr{\left[\rho_0^{[n]}\bigotimes_{i=1}^n{\left(\sigma_{\beta_i}^i\sigma_{\alpha_i}^i\sigma_{\beta_i}^i\right)}\right]}\\
        &= \left(1-\varepsilon\right)\Tr{\left[\rho_0^{[n]}O_{\{\alpha_i\}}\right]} + \frac{\varepsilon}{2^n}\Tr{\left[\rho_0^{[n]}\bigotimes_{i=1}^n{\left(\sum_{\beta_i \in \{0, 3\}}\sigma_{\beta_i}^i\sigma_{\alpha_i}^i\sigma_{\beta_i}^i\right)}\right]}.
    \end{aligned}
\end{align}
We note that

\vspace{-3mm}
\begin{align}
    \sum_{\beta_i \in \{0, 3\}}\sigma_{\beta_i}^i\sigma_{\alpha_i}^i\sigma_{\beta_i}^i = \sigma_{\alpha_i}^i + \sigma_z^i\sigma_{\alpha_i}^i\sigma_z^i = \left\{
        \begin{aligned}
            &2\sigma_{\alpha_i}^i, &\text{if }\alpha_i \in \{0, 3\},\\
            &0, &\text{if }\alpha_i \in \{1, 2\},
        \end{aligned}
        \right.
\end{align}
from which we can derive

\vspace{-3mm}
\begin{align}
    \braket{O_{\{\alpha_i\}}} = \left\{
        \begin{aligned}
            &\braket{O_{\{\alpha_i\}}}_{(0)}, &\text{if }\alpha_i\in \{0, 3\} \text{ for }\forall i,\\
            &\left(1 - \varepsilon\right)\braket{O_{\{\alpha_i\}}}_{(0)}, &\text{otherwise}.
        \end{aligned}
        \right.
\end{align}
Therefore, the length of $\vec{r}$ defined above Eq. (5) in the main text satisfies

\vspace{-3mm}
\begin{align}
    \left(1-\varepsilon\right)\leq \frac{\left|\vec{r}\left(\rho\right)\right|}{\left|\vec{r}\left(\rho_0\right)\right|} \leq 1.
\end{align}
Compared with Eq. (36) (see Methods), we conclude that the bounds in Corollary 1 can be reached for $n=1$, while for a big $n$, the physical implementability still serves as a good estimation.

\subsection{Numerical results for mixed states}
In the main text, we have calculated the decrease of logarithmic negativity $\left|\Delta E_N\right|$ for pure states $\rho_0=\ket{\psi}\hspace{-1mm}\bra{\psi}$ in Fig. 1 (a) and mixed states without positivity constraints $\rho_0 = \lambda_1\ket{\psi_1}\hspace{-1mm}\bra{\psi_1} + \lambda_2\ket{\psi_2}\hspace{-1mm}\bra{\psi_2}$ in Fig. 1 (b), where $\lambda_1$ and $\lambda_2$ are randomly chosen from $[-1, 1]$ and normalized as $\lambda_1+\lambda_2=1$.
Here we additionally provide numerical results for physical mixed states, i.e., $\lambda_1$ and $\lambda_2$ are randomly chosen from $[0, 1]$ and normalized as $\lambda_1+\lambda_2=1$, as shown in Fig. \ref{fig: Negativity}, which shows a similar distribution to that of pure states but a little smaller overall.
It means that the enlargement of the distribution range shown in Fig. 1 (b) of the main text entirely stems from the input states with negative eigenvalues.
\begin{figure}
    \includegraphics[width=0.49\linewidth]{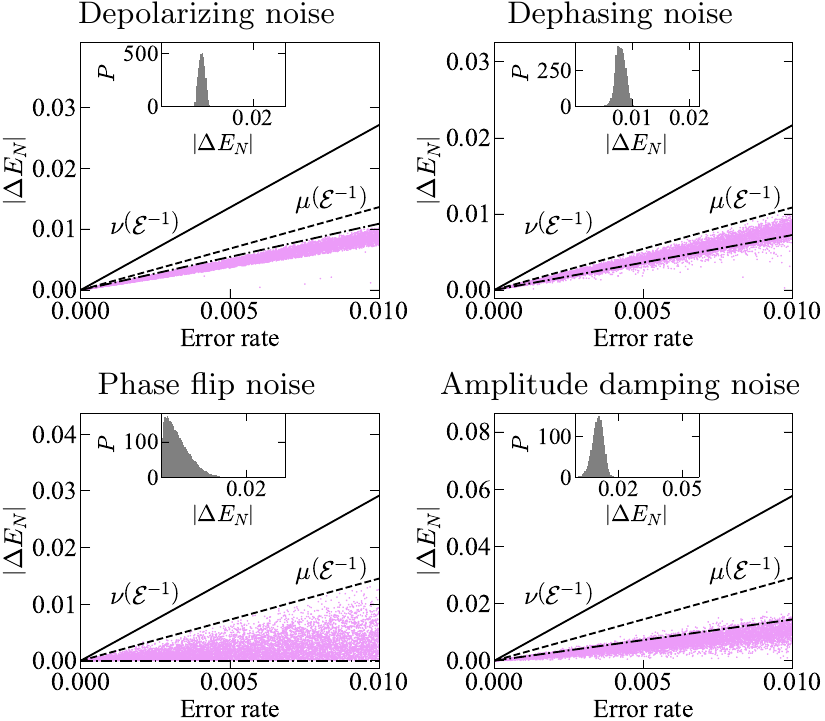}
    \caption{The change of logarithmic negativity of quantum states $\left|\Delta E_N\right|$.
    The input states are randomly chosen as mixture of two-qubit pure states $\rho_0 = \lambda_1\ket{\psi_1}\hspace{-1mm}\bra{\psi_1} + \lambda_2\ket{\psi_2}\hspace{-1mm}\bra{\psi_2}$.
    $\lambda_1$ and $\lambda_2$ are randomly chosen from $[0, 1]$ and normalized as $\lambda_1+\lambda_2=1$.
    To benchmark, we plot the upper bounds given by the physical implementability of the noise inverse $\nu{\left(\mathcal{E}^{-1}\right)}$ (solid lines), the estimation values given by root-mean-square $\mu{\left(\mathcal{E}^{-1}\right)}$ (dashed lines), and the analytical values derived for the maximally entangled state (dash-dot lines).
    Insets show the probability distributions $P$ of $\left| \Delta E_N\right|$ for error rate $\varepsilon=0.01$.}
    \label{fig: Negativity}
\end{figure}

\subsection{Decoherence effects on the maximally entangled state}
We consider the maximally entangled states undergoing noise channels.
A maximally entangled state between a bipartite system $\mathcal{H}_A\otimes \mathcal{H}_B$ is defined as

\vspace{-3mm}
\begin{align}
    \ket{\Phi}_{\max} = \frac{1}{\sqrt{d_{\rm s}}}\sum_{i}{\ket{i_Ai_B}},\label{equ:maximally_entangled}
\end{align}
where $d_A = d_B \equiv d_{\rm s}$ is the dimension of the subsystem Hilbert space.
The partial transpose of its density matrix can be written as

\vspace{-3mm}
\begin{align}
    \rho_{\max}^{{\rm T}_B} = \frac{1}{d_{\rm s}}\sum_{i,j}{\ket{i_Aj_B}\hspace{-1mm}\bra{j_Ai_B}}
\end{align}
with the corresponding eigenvalues and eigenvectors being

\vspace{-3mm}
\begin{align}
    \begin{aligned}
        &\ket{i_Ai_B}, &\lambda = \frac{1}{d_{\rm s}},\\
        &\frac{\ket{i_Aj_B}+\ket{j_Ai_B}}{2} (i\neq j), &\lambda = \frac{1}{d_{\rm s}},\\
        &\frac{\ket{i_Aj_B}-\ket{j_Ai_B}}{2} (i\neq j), &\lambda = -\frac{1}{d_{\rm s}}.
    \end{aligned}
\end{align}
Consequently, one can easily calculate the logarithmic negativity of the maximally entangled state

\vspace{-3mm}
\begin{align}
    E_N\left(\rho_{\max}\right) = \log_2{\sum_k{\left|\lambda_k\right|}} = \log_2{\left[d_{\rm s}\times\frac{1}{d_{\rm s}}+d_{\rm s}(d_{\rm s}-1)\times \frac{1}{d_{\rm s}}\right]} = \log_2{d_{\rm s}}.
\end{align}

We focus on spin-$\frac{1}{2}$ models hereafter, where the system is divided in half and we consider the bipartite entanglement.
We require that the total number of qubits $n$ are even.
The maximally entangled state on such a bipartite spin chain can be written as

\vspace{-3mm}
\begin{align}
    \ket{\Phi}_{\max}^{[n]} = \frac{1}{2^{n/4}}\left(\prod_{k=1}^{n/2}{\sum_{i_k = 0}^{1}}\right){\ket{i_1}\otimes\dots\otimes \ket{i_{n/2}}\otimes\ket{i_{n/2}}\otimes\dots\otimes\ket{i_1}}.
\end{align}
We can directly calculate the logarithmic negativity of this state $E_N\left(\rho_{\max}^{[n]}\right) = \frac{n}{2}$ if we notice that the two subsystems share $n/2$ pairs of Bell states between them.
In the following, we derive the analytical form of the negativity decrease of the maximally entangled state for four typical types of noise models, i.e., we choose $\rho_0^{[n]} = \rho_{\max}^{[n]}$.

\subsubsection{Phase flip noise}
The $n$-qubit phase flip noise is defined as

\vspace{-3mm}
\begin{align}
    \mathcal{E}^{[n]}\left(\rho^{[n]}_0\right) = \left(1-\varepsilon\right) \rho^{[n]}_0 + \varepsilon \left(\bigotimes_{k=1}^n{\sigma_{z}^k}\right)\rho^{[n]}_0\left(\bigotimes_{k=1}^n{\sigma_{z}^k}\right)^{\dagger},
\end{align}
It can be easily verified that

\vspace{-3mm}
\begin{align}
    \left(\bigotimes_{k=1}^n{\sigma_{z}^k}\right)\ket{\Phi}_{\max}^{[n]} = \ket{\Phi}_{\max}^{[n]},
\end{align}
meaning that the phase flip noise leaves the maximally entangled state unchanged, i.e., 

\vspace{-3mm}
\begin{align}
    \mathcal{E}\left(\rho_{\max}^{[n]}\right) = \rho_{\max}^{[n]}.
\end{align}
Therefore, the change of logarithmic negativity $\Delta E_N\equiv E_N\left(\rho\right) - E_N\left(\rho_0\right)$ is just equal to zero.

\subsubsection{Depolarizing noise}
The $n$-qubit depolarizing noise is defined as

\vspace{-3mm}
\begin{align}
    \mathcal{E}^{[n]}\left(\rho^{[n]}_0\right) &= \left(1-\varepsilon\right)\rho^{[n]}_0 + \frac{\varepsilon}{2^n}I^{[n]}
\end{align}
Here we adopt the more general form in Eq. \eqref{equ:maximally_entangled} and derive the output state

\vspace{-3mm}
\begin{align}
    \rho = \frac{1-\varepsilon}{d_{\rm s}}\sum_{ij}\ket{i_Ai_B}\hspace{-1mm}\bra{j_Aj_B} + \frac{\varepsilon}{d_{\rm s}^2}\sum_{ij}\ket{i_Aj_B}\hspace{-1mm}\bra{i_Aj_B}
\end{align}
with the corresponding partial transpose being

\vspace{-3mm}
\begin{align}
    \rho^{{\rm T}_B} = \frac{1-\varepsilon}{d_{\rm s}}\sum_{ij}\ket{i_Aj_B}\hspace{-1mm}\bra{j_Ai_B} + \frac{\varepsilon}{d_{\rm s}^2}\sum_{ij}\ket{i_Aj_B}\hspace{-1mm}\bra{i_Aj_B}
\end{align}
This operator is block diagonal, where the block subspaces are expanded by $\ket{i_Aj_B}$ and $\ket{i_Bj_A}$.
The eigenvalues and eigenvectors of this block diagonal matrix can be easily calculated

\vspace{-3mm}
\begin{align}
    \begin{aligned}
        &\ket{i_Ai_B}, &\lambda = \frac{1-\varepsilon}{d_{\rm s}}+\frac{\varepsilon}{d_{\rm s}^2},\\
        &\frac{\ket{i_Aj_B}+\ket{j_Ai_B}}{2} (i\neq j), &\lambda = \frac{1-\varepsilon}{d_{\rm s}}+\frac{\varepsilon}{d_{\rm s}^2},\\
        &\frac{\ket{i_Aj_B}-\ket{j_Ai_B}}{2} (i\neq j), &\lambda = -\frac{1-\varepsilon}{d_{\rm s}}+\frac{\varepsilon}{d_{\rm s}^2}.
    \end{aligned}
\end{align}
The resulting logarithmic negativity of the output state is (where we have assumed that $\frac{1-\varepsilon}{d}\geq\frac{\varepsilon}{d^2}$)

\vspace{-3mm}
\begin{align}
    E_N\left(\rho\right) = \log_2{\left[d_{\rm s}\times\left(\frac{1-\varepsilon}{d_{\rm s}}+\frac{\varepsilon}{d_{\rm s}^2}\right)+\frac{d_{\rm s}(d_{\rm s}-1)}{2}\left(\frac{1-\varepsilon}{d_{\rm s}}+\frac{\varepsilon}{d_{\rm s}^2}\right)+\frac{d_{\rm s}(d_{\rm s}-1)}{2}\left(\frac{1-\varepsilon}{d_{\rm s}}-\frac{\varepsilon}{d_{\rm s}^2}\right)\right]} = \log_2{\left[d_{\rm s}\left(1-\varepsilon\right)+\frac{\varepsilon}{d_{\rm s}}\right]},
\end{align}
from which we obtain the decrease of logarithmic negativity as 

\vspace{-3mm}
\begin{align}
    \Delta E_N = E_N\left(\rho\right)-E_N\left(\rho_0\right) = \log_2{\left[1-\left(1-\frac{1}{d_{\rm s}^2}\right)\varepsilon\right]}.
\end{align}
Here $d_{\rm s}=2^{n/2}$ in this case.

\subsubsection{Dephasing noise}
The n-qubit dephasing noise is defined as

\vspace{-3mm}
\begin{align}
    \mathcal{E}^{[n]}\left(\rho^{[n]}_0\right) = \left(1-\varepsilon\right)\rho^{[n]}_0 + \frac{\varepsilon}{2^n}\left(\prod_{k=1}^n\sum_{\alpha_k \in \{0, 3\}}\right)\left[\left(\bigotimes_{k=1}^n{\sigma_{\alpha_k}^k}\right)\rho^{[n]}_0\left(\bigotimes_{k=1}^n{\sigma_{\alpha_k}^k}\right)^{\dagger}\right].
\end{align}
We calculate a general term

\vspace{-3mm}
\begin{align}
    \begin{aligned}
        &\left(\prod_{k=1}^n\sum_{\alpha_k \in \{0, 3\}}\right)\left[\left(\bigotimes_{k=1}^n{\sigma_{\alpha_k}^k}\right)\left(\bigotimes_{k=1}^n{\ket{i_k}}\right)\left(\bigotimes_{k=1}^n{\bra{j_k}}\right)\left(\bigotimes_{k=1}^n{\sigma_{\alpha_k}^k}\right)\right]\\
        = &\left(\prod_{k=1}^n\sum_{\alpha_k \in \{0, 3\}}\right)\left[\bigotimes_{k=1}^n\left(\sigma_{\alpha_k}\ket{i_k}\hspace{-1mm}\bra{j_k}\sigma_{\alpha_k}\right)\right]\\
        = &\bigotimes_{k=1}^n\left[\sum_{\alpha_k\in \{0, 3\}}\left(\sigma_{\alpha_k}\ket{i_k}\hspace{-1mm}\bra{j_k}\sigma_{\alpha_k}\right)\right]\\
        = &2^n\prod_{k=1}^n{\delta_{i_k, j_k}}\bigotimes_{k=1}^n\ket{i_{k}}\hspace{-1mm}\bra{j_{k}}
    \end{aligned}
\end{align}
for arbitrary $\{i_k\}$ and $\{j_k\}$. In the last equality, we have used the fact that

\vspace{-3mm}
\begin{align}
    \sum_{\alpha_k\in \{0, 3\}}\left(\sigma_{\alpha_k}\ket{i_k}\hspace{-1mm}\bra{j_k}\sigma_{\alpha_k}\right) = 2\delta_{i_k, j_k}\ket{i_k}\hspace{-1mm}\bra{j_k}.
\end{align}
As a result, we can calculate the output state after the dephasing noise

\vspace{-3mm}
\begin{align}
    \rho^{[n]} = \mathcal{E}^{[n]}\left(\rho_{\max}^{[n]}\right) = \left(1-\varepsilon\right)\rho_{\max}^{[n]} + \frac{\varepsilon}{2^{n/2}}\left(\prod_{k=1}^{n/2}{\sum_{i_k = 0}^{1}}\right)\ket{i_1}\hspace{-1mm}\bra{i_1}\otimes\dots\otimes\ket{i_{n/2}}\hspace{-1mm}\bra{i_{n/2}}\otimes\ket{i_{n/2}}\hspace{-1mm}\bra{i_{n/2}}\otimes\dots\otimes\ket{i_{1}}\hspace{-1mm}\bra{i_{1}}
\end{align}
We can use the form in Eq. \eqref{equ:maximally_entangled} again to simplify our calculation

\vspace{-3mm}
\begin{align}
    \begin{aligned}
        \rho = \frac{1-\varepsilon}{d_{\rm s}}\sum_{ij}\ket{i_Ai_B}\hspace{-1mm}\bra{j_Aj_B} + \frac{\varepsilon}{d_{\rm s}}\sum_{i}\ket{i_Ai_B}\hspace{-1mm}\bra{i_Ai_B}\\
        \rho^{{\rm T}_B} = \frac{1-\varepsilon}{d_{\rm s}}\sum_{ij}\ket{i_Aj_B}\hspace{-1mm}\bra{j_Ai_B} + \frac{\varepsilon}{d_{\rm s}}\sum_{i}\ket{i_Ai_B}\hspace{-1mm}\bra{i_Ai_B}
    \end{aligned}
\end{align}
$\rho^{{\rm T}_B}$ is also a block diagonal matrix in this case, whose eigenvalues and eigenvectors are

\vspace{-3mm}
\begin{align}
    \begin{aligned}
        &\ket{i_Ai_B}, &\lambda = \frac{1-\varepsilon}{d_{\rm s}}+\frac{\varepsilon}{d_{\rm s}} = \frac{1}{d_{\rm s}},\\
        &\frac{\ket{i_Aj_B}+\ket{j_Ai_B}}{2} (i\neq j), &\lambda = \frac{1-\varepsilon}{d_{\rm s}},\\
        &\frac{\ket{i_Aj_B}-\ket{j_Ai_B}}{2} (i\neq j), &\lambda = -\frac{1-\varepsilon}{d_{\rm s}}.
    \end{aligned}
\end{align}
Therefore, the logarithmic negativity of the output state is given by

\vspace{-3mm}
\begin{align}
    E_N\left(\rho\right) &= \log_2{\left[d_{\rm s}\times\left(\frac{1}{d_{\rm s}}\right)+d_{\rm s}\left(d_{\rm s}-1\right)\frac{1-\varepsilon}{d_{\rm s}}\right]} = \log_2{\left[d_{\rm s}-\left(d_{\rm s}-1\right)\varepsilon\right]}.\\
    \Longrightarrow \Delta E_N &= \log_2{\left[1-\left(1-\frac{1}{d_{\rm s}}\right)\varepsilon\right]},
\end{align}
where $d_{\rm s} = 2^{n/2}$.

\subsubsection{Amplitude damping noise}
The single-qubit amplitude damping noise is defined via the Kraus operator $E_0 = \ket{0}\hspace{-1mm}\bra{0}+\sqrt{1-\varepsilon}\ket{1}\hspace{-1mm}\bra{1}$ and $E_1 = \sqrt{\varepsilon}\ket{0}\hspace{-1mm}\bra{1}$ with the operator-sum representation

\vspace{-3mm}
\begin{align}
    \mathcal{E}\left(\rho_0\right) = \sum_{i}E_i\rho_0E_i^{\dagger}.
\end{align}
The multi-qubit amplitude damping noise is just defined as the tensor product of the above noise channel, i.e.,

\vspace{-3mm}
\begin{align}
    \mathcal{E}^{[n]} = \bigotimes_{k=1}^n\mathcal{E}^k,
\end{align}
motivating us to consider the system as $n/2$ pairs of Bell states undergoing two-qubit noise channels independently.

\vspace{-3mm}
\begin{align}
    \begin{aligned}
        \rho^{[2]} = \mathcal{E}^{[2]}\left(\rho_{\max}^{[2]}\right) &= \frac{1}{2}[\left(1+\varepsilon^2\right)\ket{00}\hspace{-1mm}\bra{00} + \left(1-\varepsilon\right)\ket{00}\hspace{-1mm}\bra{11}+\left(1-\varepsilon\right)\ket{11}\hspace{-1mm}\bra{00}\\
        &+ \left(1-\varepsilon\right)\varepsilon\ket{01}\hspace{-1mm}\bra{01}+ \left(1-\varepsilon\right)\varepsilon\ket{10}\hspace{-1mm}\bra{10} + \left(1-\varepsilon\right)^2\ket{11}\hspace{-1mm}\bra{11}],
    \end{aligned}
\end{align}
whose partial transpose reads

\vspace{-3mm}
\begin{align}
    {\rho^{[2]}}^{{\rm T}_B} = \mathcal{E}^{[2]}\left(\rho_{\max}^{[2]}\right) &= \frac{1}{2}[\left(1+\varepsilon^2\right)\ket{00}\hspace{-1mm}\bra{00} + \left(1-\varepsilon\right)\ket{01}\hspace{-1mm}\bra{10}+\left(1-\varepsilon\right)\ket{10}\hspace{-1mm}\bra{01}\\
    &+ \left(1-\varepsilon\right)\varepsilon\ket{01}\hspace{-1mm}\bra{01}+ \left(1-\varepsilon\right)\varepsilon\ket{10}\hspace{-1mm}\bra{10} + \left(1-\varepsilon\right)^2\ket{11}\hspace{-1mm}\bra{11}].
\end{align}
It is a block diagonal matrix, and the eigenvalues and eigenvectors are

\vspace{-3mm}
\begin{align}
    \begin{aligned}
        &\ket{00}, &\lambda = \frac{1+\varepsilon^2}{2},\\
        &\ket{11}, &\lambda = \frac{\left(1-\varepsilon\right)^2}{2},\\
        &\frac{\ket{01}+\ket{10}}{2}, &\lambda = \frac{1-\varepsilon^2}{2},\\
        &\frac{\ket{01}-\ket{10}}{2}, &\lambda = -\frac{\left(1-\varepsilon\right)^2}{2}.
    \end{aligned}
\end{align}
The two subsystems share $n/2$ pairs in total, hence the logarithmic negativity of the output state is

\vspace{-3mm}
\begin{align}
    E_N\left(\rho^{[n]}\right) &= \frac{n}{2}E_N\left(\rho^{[2]}\right) = \frac{n}{2}\log_2{\left[1+\left(1-\varepsilon\right)^2\right]}.\\
    \Longrightarrow\Delta E_N &= \frac{n}{2}\log_2{\left[1-\varepsilon+\frac{\varepsilon^2}{2}\right]}.
\end{align}

\subsection{Bounds for physical implementability}
It was proved in \cite{Jiang2021} that the physical implementability for a general HPTP map is bounded by the trace norm of its Choi operator

\vspace{-3mm}
\begin{align}
    \frac{\llvert\Lambda_{\mathcal{N}}\rrvert_1}{d} \leq 2^{\nu\left(\mathcal{N}\right)} \leq \llvert\Lambda_{\mathcal{N}}\rrvert_1,\label{equ:Physical_bound}
\end{align}
where $d$ is the dimension of the Hilbert space.
In the following, we provide two lower bounds for the trace norm in terms of maximum and minimum eigenvalues.
\begin{lemma}
    For a Choi operator $\Lambda$, its trace norm satisfies
    
\vspace{-3mm}
\begin{align}
         \llvert \Lambda \rrvert_1 &\geq 2\lambda_{\max} - d,\label{equ:Norm_max}\\
         \llvert \Lambda \rrvert_1 &\geq d - 2\lambda_{\min},\label{equ:Norm_min}
    \end{align}
    where $\lambda_{\min}$ and $\lambda_{\max}$ are the minimum and maximum eigenvalues of $\Lambda$ respectively, and $d$ is the dimension of the Hilbert space.
\end{lemma}
\begin{proof}
    Suppose the eigenvalues of $\Lambda$ are $\lambda_{\min} = \lambda_0 \leq \lambda_1 \leq \dots\leq \lambda_{d^2-1}=\lambda_{\max}$ with $\Tr{\left[\Lambda\right]} = \sum_{i = 0}^{d^2-1}{\lambda_i} = d$.
    
\vspace{-3mm}
\begin{align}
        \llvert\Lambda\rrvert_1 &= \sum_{i=0}^{d^2-1}{\left|\lambda_i\right|} \geq \left|\lambda_{\max}\right| + \left|\sum_{i=0}^{d^2-2}{\lambda_i}\right| = \left|\lambda_{\max}\right| + \left|d - \lambda_{\max}\right| \geq 2\lambda_{\max} - d.\\
        \llvert\Lambda\rrvert_1 &= \sum_{i=0}^{d^2-1}{\left|\lambda_i\right|} \geq \left|\lambda_{\min}\right| + \left|\sum_{i=1}^{d^2-1}{\lambda_i}\right| = \left|\lambda_{\min}\right| + \left|d - \lambda_{\min}\right| \geq d - 2\lambda_{\min}.
    \end{align}
\vspace{-2.5mm}\end{proof}
\noindent The equality in Eq. \eqref{equ:Norm_max} holds when $\lambda_{d^2-1} \geq 0 \geq \lambda_{d^2-2}$, while the equality in Eq. \eqref{equ:Norm_min} holds when $\lambda_0 \leq 0 \leq \lambda_1$.
With Eq. \eqref{equ:Physical_bound}, we can provide two lower bounds for the physical implementability.
\begin{corollary}
    For an HPTP map $\mathcal{N}$, its physical implementability is bounded by
    
\vspace{-3mm}
\begin{align}
        2^{\nu{\left(\mathcal{N}\right)}} &\geq \frac{2\lambda_{\max}}{d} - 1,\label{equ:Physical_max}\\
        2^{\nu{\left(\mathcal{N}\right)}} &\geq 1 - \frac{2\lambda_{\min}}{d}.
    \end{align}
    where $\lambda_{\min}$ and $\lambda_{\max}$ are the minimum and maximum eigenvalues of $\Lambda_{\mathcal{N}}$ respectively, and $d$ is the dimension of the Hilbert space.
\end{corollary}
\noindent It can be verified that the bound in Eq. \eqref{equ:Physical_max} can be reached for the inverse of the multiqubit Pauli, dephasing, and depolarizing noise.

As for the possible upper bound, we note that there exists a trivial decomposition for the Choi operator of any HPTP map $\mathcal{N}$ which is not CP,

\vspace{-3mm}
\begin{align}
    \begin{aligned}
        \Lambda_{\mathcal{N}} &= \Lambda_{\mathcal{N}} - \lambda_{\min}I_{\sigma\otimes\tau} + \lambda_{\min}I_{\sigma\otimes\tau}\\
        &= \Lambda_{\mathcal{N}} - \lambda_{\min}d\frac{I_{\sigma\otimes\tau}}{d} - \left(-\lambda_{\min}\right)d\frac{I_{\sigma\otimes\tau}}{d},
    \end{aligned}
\end{align}
where $I_{\sigma\otimes\tau}$ is the identity operator defined on the extended $d^2$-dimensional Hilbert space with $\Tr_{\sigma}{\left[I_{\sigma\otimes\tau}\right]} = dI_{\tau}$, 
Therefore, $I_{\sigma\otimes \tau}/d$ is a proper Choi operator for a CPTP quantum channel, denoted as $\mathcal{T}_I$.
Actually it can be verified that $\mathcal{T}_I\left(\rho\right) = I$ for any density operator $\rho$, i.e., it projects all states to the maximally mixed state.
Consequently, we can decompose the map $\mathcal{N}$ with two CPTP maps as

\vspace{-3mm}
\begin{align}
    \mathcal{N} = \left(1-\lambda_{\min}d\right)\frac{\mathcal{N} - \lambda_{\min}d\mathcal{T}_I}{1 - \lambda_{\min}d} + \lambda_{\min}d\mathcal{T}_I,
\end{align}
from which we obtain the upper bound of the physical implementability

\vspace{-3mm}
\begin{align}    
    2^{\nu{\left(\mathcal{N}\right)}}\leq1 - 2\lambda_{\min}d.\label{equ:Physical_min}
\end{align}

For example, for the $n$-qubit depolarizing noise $\mathcal{E}^{[n]}$, the maximum and minimum eigenvalues of $\Lambda_{{\mathcal{E}^{[n]}}^{-1}}$ are

\vspace{-3mm}
\begin{align}
    \lambda_{\max} &= \frac{4^n - \varepsilon}{\left(1-\varepsilon\right)2^n},\\
    \lambda_{\min} &= -\frac{\varepsilon}{\left(1-\varepsilon\right)2^n},
\end{align}
which can be directly derived from Eq. (30) (see Methods).
Then Eq. \eqref{equ:Physical_max} and \eqref{equ:Physical_min} give

\vspace{-3mm}
\begin{align}
    \log_2\left[\frac{1+ \left(1 - \frac{2}{4^n}\right)\varepsilon}{1-\varepsilon}\right] \leq \nu{\left({\mathcal{E}^{[n]}}^{-1}\right)} \leq \log_2{\left[\frac{1+\varepsilon}{1-\varepsilon}\right]}.
\end{align}
It is implied that if we use these two bounds to estimate the physical implementability, the precision will grow exponentially with the number of qubits $n$ in this case and the estimation is exact in the thermodynamic limit.

\end{document}